\documentclass[aps,11pt,singlecolumn,prx,longbibliography,nofootinbib,superscriptaddress]{revtex4-2}

\usepackage{config}
\usepackage{comment}
\usepackage{times}
\usepackage{titlesec}

\titleformat{\paragraph}[runin]
  {\normalfont\bfseries}{\theparagraph.}{0.5em}{}[\;\;]
\titlespacing*{\paragraph}
  {0pt}{1.5ex plus 0.5ex minus 0.2ex}{0.5em}

\begin{document}

\title{Optimal query complexity for fractional quantum evolution}

\author{Anthony Yuezhang Liu}
\affiliation{Centre for Quantum Technologies, National University of Singapore, Singapore}
\affiliation{Department of Mathematics, National University of Singapore, Singapore}

\author{Adam Wesołowski}
\affiliation{Centre for Quantum Technologies, National University of Singapore, Singapore}
\affiliation{Department of Computer Science, Royal Holloway University of London, United Kingdom}
\affiliation{Department of Computer Science, University of Oxford, United Kingdom}

\author{Jayne Thompson}
\affiliation{College of Computing and Data Science, Nanyang Technological University, Singapore}
\affiliation{Nanyang Quantum Hub, School of Physical and Mathematical Sciences, Nanyang Technological University, Singapore}
\affiliation{Centre for Quantum Technologies, Nanyang Technological University, Singapore}

\author{Mile Gu}
\affiliation{Nanyang Quantum Hub, School of Physical and Mathematical Sciences, Nanyang Technological University, Singapore}
\affiliation{Centre for Quantum Technologies, Nanyang Technological University, Singapore}

\author{Lirand\"e Pira}
\email{lpira@nus.edu.sg}
\affiliation{Centre for Quantum Technologies, National University of Singapore, Singapore}

\begin{abstract}
Given oracle access to an unknown unitary \(U=e^{iH}\), the fractional query problem asks how many queries are required to implement a noninteger power \(U^t=e^{itH}\), \(0<t<1\), when the spectrum is separated from the branch cut by a gap \(\delta\). Quantum singular value transformation gives an upper bound of \(O\!\left(\frac{1}{\delta}\log\frac{1}{\varepsilon}\right)\) queries for approximation error \(\varepsilon\). We prove a matching lower bound for arbitrary query algorithms. Our argument reduces any \(N\)-query circuit to the approximation of \(e^{it\theta}\) by a trigonometric polynomial with degree bounded by \(O(N)\), together with Remez inequality. This allows us to establish the lower bound of \(\Omega_\tau\!\left(\frac{1}{\delta}\log\frac{1}{\varepsilon}\right)\). Consequently, the optimal query complexity for fractional query problem is \(\Theta_{\tau}\!\left(\frac{1}{\delta}\log\frac{1}{\varepsilon}\right)\), showing that the known QSVT construction is asymptotically optimal. We also give an alternative lower bound proof based on constructing a linear functional that annihilates the approximant space, yielding a \(\Omega_{\tau}\!\left(\log\frac{1}{\varepsilon}\right)\) bound uniform to \(\delta\).
\end{abstract}

\maketitle
\vspace{-0.3cm}

\section{Introduction}
\label{sec:intro}
 
Given a quantum circuit that implements an unknown unitary $U$, a core question is what transformations of $U$ can be implemented efficiently by an oracle query algorithm, and at what complexity. A line of work on various transformations of unknown unitaries, such as inversion, transposition, complex conjugation, and controlization, has recently produced query complexity theorems~\cite{AFCB14,QDSSM19a,QDSSM19b,YSM23,GST24,OYM25,CYZ26}. 
 
Perhaps the simplest instance of this question is the \emph{fractional query problem}. Given controlled access to $U = e^{iH}$, implement $U^t = e^{itH}$ for some $0<t<1$. Implementing integer powers are trivial, since one simply hit the oracle button for several times. However, implementing a noninteger power is much more difficult. The algorithm must extract spectral information about $U$ and transform it coherently by a function that is not generated by trivially composing. The case $t = \frac{1}{2}$, namely whether an efficient implementation of $U$ yields an efficient implementation of $\sqrt{U}$, was listed by Scott Aaronson among his ``ten most annoying questions in quantum computing''~\cite{Aar06}. 

Physically, fractional query problem asks for an efficient way to approximate continuous time quantum evolution with oracle algorithms. This viewpoint is most prominently used in the fractional query model in simulation of continuous time query algorithms and sparse Hamiltonians~\cite{CGMSY09,BCCKS14}.

Fractional powers also occupy a distinguished place among transformations of unknown unitaries. For instance, controlled access to $U$ cannot be synthesized from $U$ alone. A single query does not suffice~\cite{AFCB14}, and topological arguments show that no finite number of queries to $U$ and $U^{\dagger}$ does either, even approximately, because the global phase of the oracle is unobservable~\cite{GST24}. Fractional powers break this barrier. Dong, Nakayama, Soeda, and Murao showed that controlled-$U$, up to a phase, can be implemented exactly from queries to $d$-th roots $U^{1/d}$ for $U\in\mathrm{U}(d)$~\cite{DNSM19}. Kitaev's controlled-swap construction gives a related route when an eigenvector of $U$ is available~\cite{Kit95}, and it underlies the average case fractional power algorithm of~\cite{SMM09}. Fractional queries also mark the boundary between discrete-time and continuous-time access to a Hamiltonian. When the evolution $e^{-isH}$ is available for every $s$, rescaling time is free, and transformations of the dynamics can be studied directly as higher-order operations~\cite{OKSM24}. When only a single unitary $U=e^{iH}$ is available, the algorithm must generate the continuous time scaling by itself. One could instead try to learn $U$ first. Query-optimal unitary estimation requires $\Theta(d^{2}/\eps)$ queries to reach diamond norm error $\eps$~\cite{HaahKothariODonnellTang2023}, however, a cost that grows with the dimension and polynomially in $1/\eps$. The relevant benchmark is therefore coherent algorithms whose cost is independent of $d$ and polylogarithmic in $1/\eps$, as achieved by QSVT~\cite{GSLW19}.

In the language of quantum signal processing, the fractional query problem is about implementing spectral transformation. Writing the eigenvalues of $U$ as $e^{i\theta}$, the task is to apply the scalar map $e^{i\theta}\mapsto e^{it\theta}$ to every eigenphase. Quantum signal processing (QSP)~\cite{LYC16,LowChuang2017Optimal}, qubitization~\cite{LC19}, the quantum singular value transformation (QSVT)~\cite{GSLW19}, and generalized quantum signal processing (GQSP)~\cite{MW24} provide exactly this capability, and reduce Hamiltonian simulation, linear system solving, amplitude amplification, and phase estimation to a single task, namely, designing a polynomial that approximates a target function on the spectrum of the input~\cite{MRTC21,DLT22}. For most of these tasks, qubitization is needed first, to convert a block encoded operator into a unitary whose eigenphases carry its spectrum. In the fractional query problem the oracle is already such a unitary, and only the polynomial design remains. An $N$-query circuit realizes polynomials of degree at most $N$ in $e^{i\theta}$, but the question is how well they approximate $e^{it\theta}$. However, without further assumptions about the oracles, the fractional query implementation is impossible. The map $U \mapsto U^t$ is discontinuous over complete torus, i.e. $e^{i\theta}\mapsto e^{it\theta}$, $\theta\in[0,2\pi)$ jumps at $\theta=0$, and topological arguments show that no finite query algorithm implements it, even approximately, uniformly for all unitaries~\cite{GST24,SMM09}. The remedy~\cite{SMM09,GSLW19,MW24} is to promise that every eigenphase of $U$ stays a distance $\delta$ away from the \textit{branch cut}, i.e.\ $\operatorname{spec}(H)\subset[\delta,2\pi)$. This could be understood operationally as a requirement to be able to topologically differentiate the unitary of interest from an identity operation. Therefore, uniform approximation is possible only on an interval that excludes a gap $\delta$ around the cut, at a query complexity governed jointly by $\delta$ and the target error $\varepsilon$. Upper bounds for the fractional query problem have been obtained within these frameworks, first via QSVT with $O(\frac1\delta\log\frac1\varepsilon)$ queries~\cite{GSLW19}, and then via GQSP~\cite{MW24}.

In the past literature, Sheridan, Maslov, and Mosca~\cite{SMM09} gave an algorithm using $O\!\left(\max\{\frac{1}{\delta},\frac{1}{\eps}\}\log \frac{1}{\eps}\right)$ queries and showed that $\Omega(\frac{1}{\delta})$ queries are necessary. Gily\'en, Su, Low, and Wiebe~\cite{GSLW19} improved the error dependence, obtaining $O\!\left(\frac{1}{\delta}\log\frac{1}{\eps}\right)$ queries via QSVT. Subsequently, Motlagh and Wiebe~\cite[Theorem~10]{MW24} stated an upper bound of $O\!\left(\frac{1}{\delta}+\log \frac{1}{\eps}\right)$ using GQSP. In the same paper, the optimality with respect to the error \(\eps\), as well as the optimal interaction between two parameters, are identified as open questions. Therefore, we ask:
\begin{center}
  \emph{What is the optimal query complexity of implementing a fractional quantum evolution?}
\end{center}

In our work we resolve this question by proving a lower bound that matches the upper bound of~\cite{GSLW19}.

\subsection{Main results}

\paragraph*{Our contributions.}
We prove that for every $t\in(0,1)$ bounded away from $0$ and $1$ by \(\tau\), every query algorithm that implements $U^t$, to precision $\varepsilon$ on all oracles with \textit{spectral gap} $\delta$, takes at least $\Omega_\tau\!\left(\frac{1}{\delta}\log\frac{1}{\eps}\right)$ queries. Thereby, our result establishes the optimal query complexity of fractional queries, and resolves the open question stated in~\cite{MW24}. An immediate corollary of our lower bound is that the QSVT construction is optimal. Our main argument reduces the query complexity problem to a scalar polynomial approximation problem by constructing a hard subfamily, and we exploit the \emph{periodic mismatch obstruction} intrinsic to the problem, which by sharp Remez inequality (see~\Cref{thm:sharp-remez}) converts into a query lower bound. The formal proof can be found in~\Cref{sec:lower_bound_proof}. We also give a different lower bound proof by constructing a \textit{frequency annihilator} functional that annihilates the approximant space, yielding a \(\Omega_{\tau}\!\left(\log\frac{1}{\varepsilon}\right)\) bound uniform to \(\delta\).

Matching lower bounds are uncommon in this area. The polynomial method~\cite{Beals2001Quantum} shows that the amplitudes of any $N$-query algorithms are polynomials of degree at most $N$ in the entries of the oracle. With controlled-$U^{\pm1}$ queries, these are the entries of $U$ and their complex conjugates. Tight bounds of this kind are known for a few other central tasks, including Hamiltonian simulation~\cite{BACS07,BCK15}, phase estimation~\cite{MdW23}, and certain classes of matrix functions~\cite{MS24}. For fractional query problem, the only known lower bound is the inverse \textit{spectral gap} bound $\Omega(\frac{1}{\delta})$~\cite{SMM09}. It comes from a hybrid argument that is blind to the approximation error, and whether the $\log(\frac{1}{\eps})$ term in the GQSP upper bound is necessary was left open~\cite{MW24}.
 
\paragraph*{Consequences for the additive upper bound of generalized quantum signal processing.}
Our main theorem shows that no algorithm can achieve  $O\!\left(\frac{1}{\delta}+\log\frac{1}{\eps}\right)$
queries. The construction of~\cite[Theorem~10]{MW24}, which claimed that upper bound, is correct but the the following complexity analysis shows that it matches our lower bound. The construction first performs phase
estimation to precision $\delta$ with $O(\frac{1}{\delta})$ queries in order to decide on which side of the
branch cut each eigenphase lies. With this query budget the probability of misclassifying an
eigenphase remains bounded below by a constant. Reducing it to $\varepsilon$ requires an additional
factor of $\log\frac{1}{\eps}$, which is known to be necessary for phase estimation~\cite{MdW23}.
With this factor restored, the construction of~\cite{MW24} recovers the
$O\!\left(\frac{1}{\delta}\log\frac{1}{\eps}\right)$ scaling of QSVT.

\subsection{Proof overview}
\label{sec:intro-overview}
 
\paragraph*{Reduction to scalar polynomial approximation.}
Since the lower bound must hold uniformly over the oracle family, it suffices to construct a hard subfamily. We constructed the two-eigenvalue oracles $U_\theta = e^{i\delta}(I-P) + e^{i\theta}P$ with $\theta\in[\delta,2\pi)$ where $P$ is a rank-one projector. Each controlled query is affine in $z=e^{i\theta}$ or $z^{-1}$, so an $N$-query circuit is an operator valued Laurent polynomial of degree at most $N$ in $z$. The unknown global phase is removed by taking the product of the matrix elements on the two eigenvectors, one of them conjugated, which is a matrix element of the induced channel. This yields a trigonometric polynomial $q$ in \(\theta\) of degree at most $2N$ with $|q(\theta)-e^{it\theta}|\le 2\varepsilon$ on $[\delta,2\pi)$. This step are in the same lineage of idea in various polynomial methods~\cite{Beals2001Quantum,BunKothariThaler2020Polynomial}, adapted to the phase insensitive distance of~\cite{HaahKothariODonnellTang2023}.
 
\paragraph*{Periodic mismatch obstruction.}
The polynomial $q$ is $2\pi$-periodic, whereas $e^{it\theta}$ is not. At the branch cut the target jumps by $2\sin(\pi t)$. To exploit this we compare $q$ with a translation of itself. The function $r(\theta) = q(\theta+\delta) - e^{it\delta}q(\theta)$ is again a trigonometric polynomial of degree at most $2N$. Because $e^{it\theta}$ is a character, $r$ is at most $4\varepsilon$ in magnitude on a set of measure $2\pi-2\delta$. Periodicity nevertheless forces $|r(0)| \ge 2\sin(\pi t)-4\varepsilon$. A trigonometric polynomial that is small on most of the circle cannot be large elsewhere unless its degree is large. The sharp Remez inequality of Tikhonov and Yuditskii~\cite{TikhonovYuditskii2020SharpRemez} makes this precise with the optimal constant $T_{4N}(\sec(\delta/2))\le e^{4N\delta}$. The periodic mismatch obstruction was introduced in~\cite{LiuPira2026} to bound the depth-error tradeoff of fixed data re-uploading circuits. Here, it is used to bound the query complexity.

\paragraph*{Organization.}
Section~\ref{sec:problem} formalizes the fractional query problem, the query model and the success criterion, and sets the notation. Section~\ref{sec:lower} gives the reduction, the joint lower bound, and the frequency annihilator bound, and combines them into the optimality theorem. Section~\ref{sec:implication} discusses implications and open problems. Proofs can be found in
the appendices.

\section{Fractional queries}
\label{sec:problem}

\subsection*{Notation}

We write $\mathbb N=\{0,1,\ldots\}$ and $\mathbb T=\{z\in\mathbb C:|z|=1\}$. All logarithms are natural. For an operator $X$, $X^\dagger$, $\spec(X)$, and $\Tr(X)$ denote its adjoint, spectrum, and trace, respectively. We use $\|\cdot\|_2$ for the Hilbert-space norm, $\|\cdot\|_{\mathrm{op}}$ for the operator norm, $\|\cdot\|_1$ for the trace norm, and $\|\cdot\|_\diamond$ for the diamond norm of linear maps. The identity is denoted by $Id$ or $\mathbbm{1}$; register labels and identity tensor factors are omitted when clear. Induced channel \(\cE_V\) of a unitary \(V\) is defined by \(\mathcal E_V(\rho)=V\rho V^\dagger\). 

For a compact interval $I\subset\mathbb R$, $C(I)$ denotes the complex-valued continuous functions on $I$, equipped with $\|f\|_{\infty,I}:=\sup_{x\in I}|f(x)|$. We omit $I$ when the domain is clear and write $\|\phi\|$ for the corresponding dual norm of a bounded linear functional. For sets, $|\cdot|$ denotes cardinality or Lebesgue measure, as appropriate. We use $O$, $\Omega$, and $\Theta$ in their standard senses in complexity theory; subscripts indicate the only parameters on which implicit constants may depend.

\subsection{Query model}

We consider a query model that allows access to controlled-\(U\), controlled-\(U^{\dagger}\), and arbitrary additional ancillary workspace. It is equivalent to a quantum comb as defined in \cite{Chiribella2008QuantumComb}. The controlled query setting follows the conventions and of \cite{SMM09,GSLW19,MW24}.

\begin{definition}[\(N\)-query circuit]
\label{def:model}

Let \(C\) be a control qubit and let the controlled query be
\begin{equation}
  \mathcal O_U^{(\sigma)}
  =
  |0\rangle\!\langle0|_C\otimes U^{\sigma}
  +
  |1\rangle\!\langle1|_C\otimes I_S,
  \label{eq:oracle}
\end{equation}
with \(\sigma\in\{-1,+1\}\), extended by the identity on all additional registers. An \(N\)-query circuit has the form
\begin{equation}
  \mathcal A_U
  =
  W_N\mathcal O_U^{(\sigma_N)} W_{N-1}\cdots W_1\mathcal O_U^{(\sigma_1)} W_0, \qquad \sigma_j \in \{-1,+1\}
  \label{eq:n-query-circuit}
\end{equation}
where the unitaries \(W_0,\ldots,W_N\) are independent of \(U\). Figure~\ref{fig:quantum-comb-model} illustrates our query model.

\end{definition}

More discussion and detailed treatment of the model is in Appendix~\ref{app:model}. 

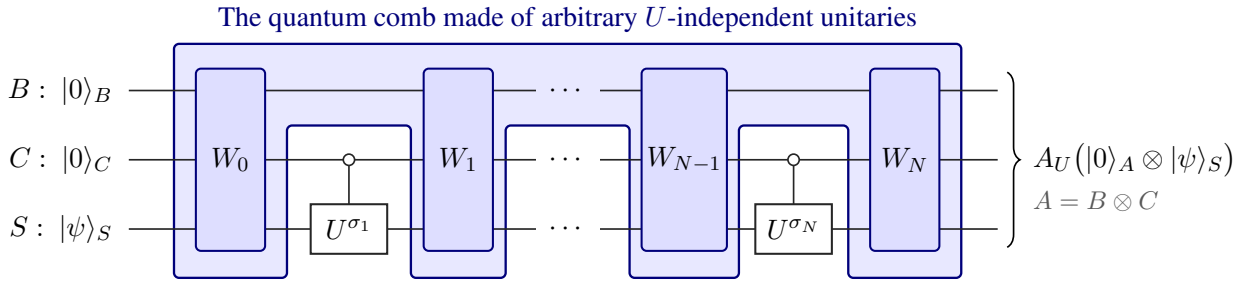
\begin{figure}[htbp]
\label{fig:1}
\centering
\resizebox{\linewidth}{!}{
\begin{tikzpicture}[
  x=1cm,y=1cm,
  wire/.style={draw=black!85,line width=.65pt},
  tooth/.style={draw=blue!48!black,fill=blue!13,
    line width=.8pt,rounded corners=2pt,
    minimum width=9.5mm,minimum height=25mm,inner sep=2pt},
  query/.style={draw=black!85,fill=white,line width=.8pt,
    minimum width=10.5mm,minimum height=6.8mm,inner sep=2pt},
  every node/.style={font=\normalsize}
]

\path[draw=blue!48!black,fill=blue!9,line width=.9pt,
      rounded corners=2pt]
  (0,2.14) -- (10.8,2.14) -- (10.8,-1.07)
  -- (9.25,-1.07) -- (9.25,1.02)
  -- (7.75,1.02) -- (7.75,-1.07)
  -- (6.25,-1.07) -- (6.25,1.02)
  -- (4.55,1.02) -- (4.55,-1.07)
  -- (3.25,-1.07) -- (3.25,1.02)
  -- (1.55,1.02) -- (1.55,-1.07)
  -- (0,-1.07) -- cycle;

\node[text=blue!48!black] at (5.4,2.48)
  {The quantum comb made of arbitrary $U$-independent unitaries};

\node[anchor=east] at (-.72,1.5) {$B:\ |0\rangle_B$};
\node[anchor=east] at (-.72,.55) {$C:\ |0\rangle_C$};
\node[anchor=east] at (-.72,-.4) {$S:\ |\psi\rangle_S$};
\foreach \y in {1.5,.55,-.4}{
  \draw[wire] (-.62,\y) -- (11.30,\y);
}

\node[tooth] at (.775,.55) {$W_0$};
\node[tooth] at (3.9,.55) {$W_1$};
\node[tooth] at (7.0,.55) {$W_{N-1}$};
\node[tooth] at (10.025,.55) {$W_N$};

\foreach \x/\index in {2.4/1,8.5/N}{
  \draw[wire] (\x,.55) -- (\x,-.4);
  \node[query] at (\x,-.4) {$U^{\sigma_{\index}}$};
  \draw[wire,fill=white] (\x,.55) circle[radius=2.2pt];
}

\node[fill=blue!9,inner xsep=5pt,inner ysep=0pt] at (5.4,1.5)
  {$\cdots$};
\foreach \y in {.55,-.4}{
  \node[fill=white,inner xsep=5pt,inner ysep=0pt] at (5.4,\y)
    {$\cdots$};
}

\draw[decorate,decoration={brace,amplitude=5pt},line width=.65pt]
  (11.43,1.75) -- (11.43,-.65);
\node[anchor=west,inner sep=0pt] at (11.78,.55)
  {$A_U\bigl(|0\rangle_A\otimes|\psi\rangle_S\bigr)$};
\node[anchor=west,font=\small,text=black!65,inner sep=0pt]
  at (11.78,-.02) {$A=B\otimes C$};
\end{tikzpicture}%
}
\caption{An $N$-query unitary as a realization of a quantum comb, following Chiribella, D'Ariano and Perinotti \cite{Chiribella2008QuantumComb}. The connected blue veil contains the arbitrary $U$-independent unitaries $W_0,\ldots,W_N$ and the memory wire $B$; the $N$ slots are filled by the supplied controlled queries on $C\otimes S$.}
\label{fig:quantum-comb-model}
\end{figure}

\subsection{Fractional query problem}
\label{sec:frac}
Two obstructions determine the formulation of the fractional query problem. First, for \(0<t<1\), no finite query circuit can implement \(U^t\) exactly and uniformly over a family of oracles. We therefore allow an approximation error \(\eps>0\) and study how to implement \(\eps\)-approximately. Second, a well defined continuous fractional power map requires a choice of spectral branch. On the branch \(0\leq\theta<2\pi\), the eigenvalues \(e^{i\theta}\) approach the same limit as \(\theta\to0^+\) and \(\theta\to2\pi^-\), whereas their fractional powers approach \(1\) and \(e^{2\pi it}\), respectively. The discontinuity of the \((\cdot)^t\) map at the endpoint of the branch cut prevents arbitrarily accurate uniform approximation over any oracle unitaries. A promised spectral branch cut gap parameter \(\delta\) is therefore necessary. Similar argument can also be found in \cite{SMM09}, where an \(\Omega_\tau(\frac1\delta)\) lower bound is shown, where \(\tau\) is a bound of the parameter \(t\) away from \(0\) and \(1\). 

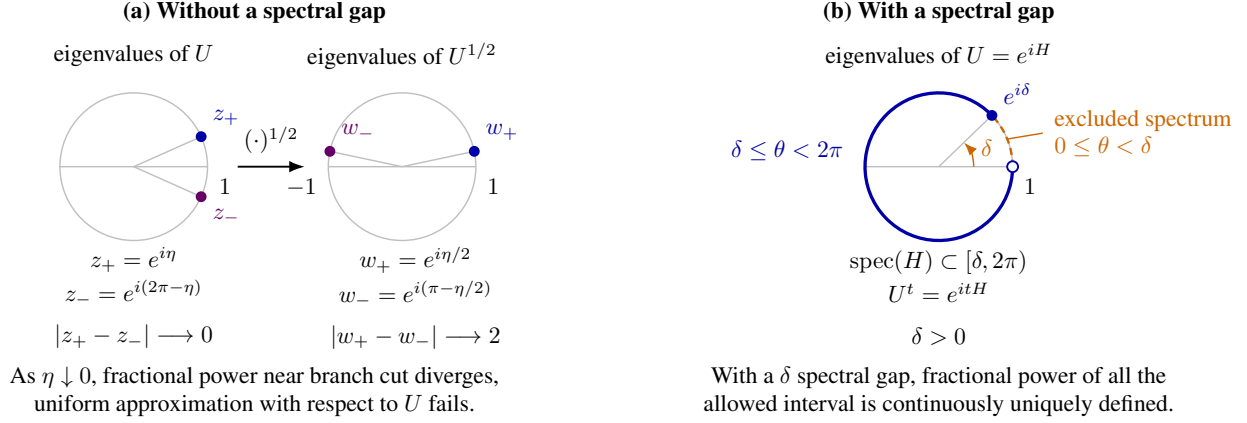
\begin{figure}[htbp]
\centering
\resizebox{\ifdim\width>\linewidth\linewidth\else\width\fi}{!}{%
\begin{tikzpicture}[
  font=\small, >=Latex,
  guide/.style={draw=black!25,line width=.6pt},
  allowed/.style={blue!65!black,line width=1.4pt},
  gap/.style={orange!80!black,densely dashed,line width=1pt},
  upper/.style={blue!65!black},
  lower/.style={violet!80!black},
  heading/.style={font=\small\bfseries},
  formula/.style={align=center},
  note/.style={font=\small,align=center}
]
\def\r{1.1}
\def\dotsize{2.3pt}

\node[heading] at (1.8,2.3) {(a) Without a spectral gap};
\node[heading] at (12,2.3) {(b) With a spectral gap};

\node at (0,1.65) {eigenvalues of $U$};
\draw[guide] (0,0) circle (\r);
\draw[guide] (-\r,0)--(\r,0);
\draw[guide] (24:\r)--(0,0)--(-24:\r);
\fill[upper] (24:\r) circle (\dotsize);
\fill[lower] (-24:\r) circle (\dotsize);
\node[upper,above right=2pt] at (24:\r) {$z_+$};
\node[lower,below right=2pt] at (-24:\r) {$z_-$};
\node[below right=2pt] at (0:\r) {$1$};

\draw[->,line width=.8pt] (1.55,0)--(2.55,0)
  node[midway,above=2pt] {$(\cdot)^{1/2}$};

\begin{scope}[xshift=4cm]
  \node at (0,1.65) {eigenvalues of $U^{1/2}$};
  \draw[guide] (0,0) circle (\r);
  \draw[guide] (-\r,0)--(\r,0);
  \draw[guide] (12:\r)--(0,0)--(168:\r);
  \fill[upper] (12:\r) circle (\dotsize);
  \fill[lower] (168:\r) circle (\dotsize);
  \node[upper,above right=2pt] at (12:\r) {$w_+$};
  \node[lower,above right=2pt] at (168:\r) {$w_-$};
  \node[below right=2pt] at (0:\r) {$1$};
  \node[below left=2pt] at (180:\r) {$-1$};
\end{scope}

\node[formula] at (0,-1.95)
  {$z_+=e^{i\eta}$\\[3pt]
   $z_-=e^{i(2\pi-\eta)}$\\[6pt]
   $|z_+-z_-|\longrightarrow0$};
\node[formula] at (4.2,-1.95)
  {$w_+=e^{i\eta/2}$\\[3pt]
   $w_-=e^{i(\pi-\eta/2)}$\\[6pt]
   $|w_+-w_-|\longrightarrow2$};
\node[note] at (1.8,-3.35)
  {As $\eta\downarrow0$, fractional power near branch cut diverges, \\uniform approximation with respect to \(U\) fails.};

\begin{scope}[xshift=12cm]
  \node at (0,1.65) {eigenvalues of $U=e^{iH}$};
  \draw[guide] (0,0) circle (\r);
  \draw[guide] (-\r,0)--(\r,0);
  \draw[guide] (0,0)--(44:\r);
  \draw[allowed]
    (44:\r) arc[start angle=44,end angle=360,radius=\r];
  \draw[gap]
    (0:\r) arc[start angle=0,end angle=44,radius=\r];
  \draw[->,orange!80!black,line width=.6pt]
    (0:.52) arc[start angle=0,end angle=44,radius=.52];
  \node[orange!80!black] at (22:.77) {$\delta$};
  \fill[upper] (44:\r) circle (\dotsize);
  \draw[upper,fill=white,line width=.8pt]
    (0:\r) circle (\dotsize);
  \node[upper,above right=2pt] at (44:\r) {$e^{i\delta}$};
  \node[below right=2pt] at (0:\r) {$1$};
  \node[upper,anchor=east,align=right] at (-1.3,.25)
    {$\delta\leq\theta<2\pi$};
  \draw[orange!80!black,line width=.6pt]
    (22:\r)--(1.55,.5);
  \node[orange!80!black,anchor=west,align=left] at (1.6,.5)
    {excluded spectrum\\$0\leq\theta<\delta$};

  \node[formula] at (0,-1.95)
    {$\operatorname{spec}(H)\subset[\delta,2\pi)$\\[3pt]
     $U^t=e^{itH}$\\[6pt]
     $\delta>0$};
  \node[note] at (0,-3.35)
    {With a \(\delta\) spectral gap, fractional power of all the\\ allowed interval is continuously uniquely defined.};
\end{scope}
\end{tikzpicture}
}
\caption{Why a spectral gap is needed for uniform fractional power
approximation. On the branch $0\leq\theta<2\pi$, eigenvalues approaching
$1$ from opposite sides become arbitrarily close, while their square
roots approach $1$ and $-1$. More generally, for $0<t<1$, the two limiting
fractional powers differ by $|1-e^{2\pi it}|=2\sin(\pi t)>0$.
The promise $\operatorname{spec}(H)\subset[\delta,2\pi)$ excludes
$0\leq\theta<\delta$, removing one of these approaches to the cut.
The solid blue arc is allowed, $\theta=\delta$ is included and
$\theta=2\pi$ is excluded.}
\label{fig:branch-margin-definition}
\end{figure}

To formulate the problem precisely, we follow the spectral convention of \cite[Theorem~10]{MW24}. We write \(U=e^{iH}\), where \(H=H^\dagger\), and impose \(\spec(H)\subset[\delta,2\pi)\). The target is then \(U^t=e^{itH}\). This spectral gap excludes the spectral interval \([0,\delta)\), thus removing the domain where discontinuity can happen, as illustrated in Figure~\ref{fig:branch-margin-definition}. Therefore, for fixed \(t\), the problem naturally has two complexity parameters. The approximation error \(\eps\) and the spectral gap \(\delta\). We require a single circuit, whose inter query unitaries are independent of \(U\), to achieve the prescribed accuracy uniformly over all oracles satisfying this promise. The following definitions make the task and its success criterion precise.

\begin{definition}[Fractional query problem, adapted from {\cite[Theorem~10]{MW24}}]
\label{def:fractional-query}
Let \(\delta,t,\eps\in(0,1)\), and let \(U=e^{iH}\), where \(H=H^\dagger\) and
\begin{equation}
  \spec(H)\subset[\delta,2\pi).
  \label{eq:gqsp-gap-promise}
\end{equation}
Given access to controlled-\(U^\sigma\) for \(\sigma\in\{-1,+1\}\), arbitrary finite ancillas, and oracle independent unitaries, the task is to implement an \(\eps\)-approximation of
\begin{equation}
  U^t=e^{itH},
  \label{eq:gqsp-target}
\end{equation}
the fractional power of \(U\). 
\end{definition}

Now we define the success criterion of a fractional query algorithm. Given the parameters \((t,\delta,\eps)\) that specifies the setting and the task, a successful algorithm should implement \(\eps\)-approximation of fractional power of queries uniform over all possible oracle unitaries in the family defined in Definition~\ref{def:fractional-query}.

\begin{definition}[\((t,\delta,\eps)\)-success and query complexity]\label{def:diamondsuccess}
A circuit is $(t,\delta,\eps)$-successful in fractional query problem if for the induced channel \(\cE_{\cA_U}\) of \(\cA_U\) and \(\cE_{U^t}\) of \(U^t\), where we isometrically embedding \(U^t\) into the workspace of \(A_U\),
\begin{equation}
 \sup_{\substack{U=e^{iH}, H=H^\dagger\\\spec(H)\subset[\delta,2\pi)}}
  \dnorm{\cE_{\cA_U}-\cE_{U^t}}\leq\eps.
 \label{eq:diamondsuccess-condition}
\end{equation}
The query complexity $Q(t,\delta,\eps)$ is the smallest $N\in\N$ for which such an \(N\)-query circuit exists. In other words, we ask for \(\eps\)-approximation up to global phase.
\end{definition}

More discussion and detailed treatment of the success criterion is in Appendix~\ref{app:successcrit}

\subsection{Known QSVT upper bound}

\begin{theorem}[QSVT joint upper bound, from \cite{GSLW19}]
\label{thm:qsvt-upper}
For every \(t,\delta\in(0,1)\) and \(\eps \in (0,\frac12)\), there exists a \((t,\delta,\eps)\)-successful circuit using
\[
    O\!\left(\frac{1}{\delta}\log\frac{1}{\eps}\right)
\]
controlled-\(U\) queries. In particular, for a universal constant \(C>0\),
\begin{equation}
    Q(t,\delta,\eps)
    \leq C\!\left(\frac{1}{\delta}\log\frac{1}{\eps}\right).
    \label{eq:mw-upper}
\end{equation}
\end{theorem}

\section{Optimal query complexity}
\label{sec:lower}

In this section, we introduce our main results.

\begin{theorem}[Joint lower bound of query complexity]
\label{thm:remez-lower}
Fix \(\tau\in(0,1/2)\). For all \((t,\delta,\eps)\) with \(t\in[\tau,1-\tau]\), \(\eps\in(0,\left(\frac{\sin(\pi\tau)}{4}\right)^2)\), and \(\delta\in(0,1)\), 
\begin{equation}
Q(t,\delta,\eps)
=\Omega_{\tau}\left(\frac{1}{\delta}log\frac{1}{\eps}\right)
\end{equation}
Specifically,
\begin{equation}
    Q(t,\delta,\eps)
    \geq\frac{1}{8\delta}\log\frac{1}{\eps}.
    \label{eq:remez-uniform-lower}
\end{equation}
\end{theorem}

The proof is given in Appendix~\ref{proof:remez}.

\begin{remark}[Dependence on \(\tau\)]
The dependence on \(\tau\) is necessary. If we don't assume that \(t\) is bounded away from \(0\) and \(1\) by \(\tau\), at \(t=0\), the target is the identity and at \(t=1\), one controlled-\(U\) query implements it exactly. And a single query approximates $U^{t}$ to error $O(\min\{t,1-t\})$. Uniform lower bound over all possible oracles can hold only when \(t\) is bounded away from the endpoints.
\end{remark}

As a consequence, we've shown that the QSVT upper bound construction provided by Theorem~\ref{thm:qsvt-upper} from \cite{GSLW19} is optimal. Hence, we show an optimal joint query complexity for fractional query problem in the parameters \(\eps\) and \(\delta\), depending on a assumption of \(\tau\) bounding the value of power parameter \(t\).

\begin{theorem}[Optimal query complexity for the fractional query problem]
\label{thm:optimality}
Fix $\tau\in(0,1/2)$. There exists a universal constant $C>0$ such that, for all \((t,\delta,\eps)\) with \(t\in[\tau,1-\tau]\), \(\eps\in(0,\left(\frac{\sin(\pi\tau)}{4}\right)^2)\), and \(\delta\in(0,1)\),
\begin{equation}
    Q(t,\delta,\eps)
    =\Theta_\tau\!\left(
        \frac{1}{\delta}\log\frac{1}{\eps}
    \right)
\end{equation}
Specifically,
\begin{equation}
    \frac{1}{8\delta}\log\frac{1}{\eps}
    \leq Q(t,\delta,\eps)
    \leq \frac{C}{\delta}\log\frac{1}{\eps}.
    \label{eq:optimal-query-complexity}
\end{equation}
Thus the QSVT construction is optimal up to constant factors.
\end{theorem}

The proof of this theorem is given in Appendix~\ref{app:main}.

\begin{remark}[GQSP additive bound]

In the work of GQSP \cite{MW24}, we've noticed a claimed upper bound of \(O\!\left(\frac{1}{\delta}+log\frac{1}{\eps}\right)\!\). We believe the construction to be correct, and appropriate complexity analysis of that construction reveals that the approach matches our lower bound. 
    
\end{remark}

\section{An alternative lower bound via frequency annihilation}

We also give an alternative lower bound with respect to only \(\eps\), uniform to \(\delta\in(0,1)\) using a different proof method.

\begin{theorem}[\(\eps\) lower bound by constructing frequency annihilator]
\label{thm:fd-lower}Fix $\tau\in(0,1/2)$. For all \((t,\delta,\eps)\) with \(t\in[\tau,1-\tau]\), \(\eps\in(0,(\frac{\tau(1-\tau)}2)^2)\)and \(\delta\in(0,1)\),

\begin{equation}
Q(t,\delta,\eps)
=\Omega_{\tau}\left(log\frac{1}{\eps}\right)
\end{equation}
More specifically,
\begin{equation}
    Q(t,\delta,\eps)
    \geq
    \max\!\left\{
        0,
        \left\lceil
            \frac{\log(t(1-t)/(2\eps))}{4\log(6e)}
        \right\rceil
    \right\}.
    \label{eq:fd-query-lower}
\end{equation}
\end{theorem}

The proof of this theorem if given in Apppendix~\ref{proof:freqlower}

\section{Discussion and outlook}
\label{sec:implication}

We showed that $Q(t,\delta,\varepsilon)=\Theta_{\tau}\!\left(\frac{1}{\delta}\log\frac{1}{\varepsilon}\right)$. The QSVT construction of~\cite{GSLW19} is therefore optimal up to constant factors. Several questions remain open.

\paragraph*{Dependence on $t$.}
Our bound vanishes once $\varepsilon\gtrsim\sin(\pi t)$. This is consistent with the fact that the identity, or a single query, approximates $U^{t}$ to error $O(\min\{t,1-t\})$ up to a global phase. We conjecture that $Q=\Theta\!\left(\frac{1}{\delta}\log\frac{\sin(\pi t)}{\varepsilon}\right)$ uniformly for $\varepsilon\le c\sin(\pi t)$. Proving this requires an upper bound whose cost adapts to the size of the jump.

\paragraph*{Uncontrolled access.}
Queries to $U$ and $U^{\dagger}$ alone are a special case of our model, so the lower bound applies to them. It is unclear whether the matching upper bound survives without controlled access. Controlled-$U$ cannot be synthesized from $U$~\cite{AFCB14,GST24}, and the known control-free constructions require an eigenvector ancilla or hold only on average~\cite{SMM09,Kit95}.

\paragraph*{Beyond the worst case.}
The spectral gap is a worst case promise. A Haar-random $U\in\mathrm{U}(d)$ typically has an eigenphase within $O(1/d)$ of the branch cut, and the average-case algorithm of~\cite{SMM09} needs no promise at all. The optimal average-case query complexity as a function of $d$ and $\varepsilon$ remains open. Our hard instance has only two eigenvalues and says nothing about typical oracles.

\paragraph*{Other spectral transformations.}
The shift argument applies whenever the target obeys an affine relation $f(\theta+\delta)=a\,f(\theta)+b$ on the promised arc that its periodic extension violates. A natural conjecture is a cost of $\Theta\!\left(\frac{1}{\delta}\log\frac{J}{\varepsilon}\right)$ whenever $f$ is analytic on the arc, where $J$ is the jump of $e^{if}$ across the gap. Establishing this would give a unified lower bound theory for phase functions of unitary oracles.

\section*{Acknowledgments} This work is supported by A*STAR under its Young Investigator Research Grant (YIRG) M25N8c0131; the National Research Foundation, Singapore through the National Quantum Office, hosted in A*STAR, under its Centre for Quantum Technologies Funding Initiative (S24Q2d0009); the RIE 2025 AQAS projects S25Q9D001 and S25Q9D002; the National Research Foundation of Singapore through the NRF Investigatorship Program (Award No. NRF-NRFI09-0010).

\emph{AI Declaration} We have used ChatGPT 5.6 and 6.0 to proofread all the proofs in the paper. Several elementary estimation steps in the proofs were proposed by AI and examined before applying. AI is also used for the improvement of presentation. The authors assume full responsibility for the correctness of the manuscript.

\appendix

\section{Details about query model}
\label{app:model}

The query model we consider in Definition~\ref{def:model} allows controlled-\(U\), controlled-\(U^{\dagger}\), and arbitrary additional ancillary workspace. We fix the control bit at register \(C\) for simplicity, whose space is denoted as \(\mathcal{H_C}\). Any controlled-\(U\) query whose logical control is another ancillary qubit \(C_j\) can be written as
\(
  \mathcal O_U^{(C_j)}
  =
  S_j^\dagger \mathcal O_U^{(C)} S_j,
\)
where \(S_j\) swaps \(C\) and \(C_j\). Since \(S_j\) is independent of \(U\), these swaps may be absorbed into the adjacent inter query unitaries. We fix all ancillary bits at the cluster of registers \(B\), whose space is denoted as \(\mathcal{H_B}\). Measurements and classical controls can be carried out by dilating the operation using workspace in \(B\). The original hilbert space on which \(U\) acts will be denoted by \(S\), and the space will be denoted as \(\mathcal{H_S}\). Additionally, for notational frugality, we define the space of all extra registers as \(A\) and $\mathcal H_A=\mathcal H_B\otimes\mathcal H_C$, where $C$ is the control qubit and $B$ contains the remaining ancillary workspace.
 
An \(N\)-query circuit in this model has a quantum comb representation as defined in \cite{Chiribella2008QuantumComb}. Removing the oracle calls leaves an oracle independent circuit with \(N\) open slots arranged in a fixed order. Each slot has input and output space \(\mathcal H_C\otimes\mathcal H_S\), while the ancillary workspace \(B\) carries quantum memory between successive slots. The inter-query unitaries connect these slots, and inserting the channels induced by controlled-\(U\) or controlled-\(U^\dagger\) recovers the query algorithm. This network is causal, since the operation inserted into a later slot cannot affect the input supplied to an earlier slot. Moreover, inserting arbitrary completely positive trace preserving maps into the slots produces a completely positive trace preserving map, so the underlying network is a deterministic quantum comb. Measurements and classical feedback are included through their coherent dilations, with the corresponding outcome records retained in \(B\).

\section{Details about success criterion}
\label{app:successcrit}

Our main problem is about how to implement an approximation. In order to address approximation meaningfully, we first define a notion of distance that we work on. The metric induced by ordinary operator norm does not satisfy our need because they do not count for the equivalence up to global phases. Therefore, we adapt the diamond norm for the definition of the \(\eps\)-approximation. It retains the operational meaning we are interested in~\cite{Regula2021operational}. 

However, for mathematical consistency with our query model, also for avoiding extra formalism introduced by channels, throughout the proofs we use an equivalent notion of distance given by Haah, Kothari, O'Donnell, and Tang in \cite{HaahKothariODonnellTang2023}. In the same work they have shown that this distance is equivalent to diamond norm of the induced channel up to universal constant. We use the direct isometric extension for the convenience in handling ancillas. 

\begin{definition}[Phase insensitive distance, adapted from \cite{HaahKothariODonnellTang2023}]
For isometries \(V,W:\mathcal K\to\mathcal L\), define
\begin{equation}
    d_{ph}(V,W)
    :=
    \min_{\omega\in\mathbb T}\opnorm{V-\omega W}.
    \label{eq:phase-distance}
\end{equation}
The minimum is attained because \(\mathbb T\) is compact.
\end{definition}

This distance forms a pseudometric, it's associated to the diamond norm of the induced channel. 

\begin{lemma}[Relation with diamond norm of the induced channel, \cite{HaahKothariODonnellTang2023}]
\label{lem:relationdiamond}
For isometries $V,W:\mathcal K\to\mathcal L$, define the induced channels $\mathcal E_V(\rho)=V\rho V^\dagger$ and $\mathcal E_W(\rho)=W\rho W^\dagger$. The unitary comparison in \cite[Proposition~1.6]{HaahKothariODonnellTang2023} extends to isometries:
\[
    \frac12\dnorm{\mathcal E_V-\mathcal E_W}
    \leq d_{ph}(V,W)
    \leq \dnorm{\mathcal E_V-\mathcal E_W}.
\]
These two distance are equivalent up to universal constants.

\end{lemma}

We define the isometry \(J\) that records the embedding of the original state into the space where the algorithm work.
\begin{equation}
 J:\cH_S\longrightarrow\cH_{\mathrm{A}}\otimes\cH_S,
 \qquad J\ket\psi:=\ket0_{A}\otimes\ket\psi,
 \label{eq:embedding}
\end{equation}
where $\ket0_A$ is a fixed normalized ancilla state. The map $\cA_UJ$ is the output, and $JU^t$ is the desired approximation target. Comparing them requires the ancillas to return approximately to their initialized state, as well as requiring the correct action on $S$.

\begin{figure}[htbp]
\centering

\begin{quantikz}[
    row sep={0.85cm,between origins},
    column sep=0.65cm
]
    & \gate[wires=2]{J}\wireoverride{n}
    & \rstick{$\ket{0}_A$}
    \\
    \lstick{$\ket{\psi}_S$}
    &
    & \rstick{$\ket{\psi}_S$}
\end{quantikz}
\quad$\equiv$\quad

\begin{quantikz}[
    row sep={0.85cm,between origins},
    column sep=0.65cm
]
    \lstick{$\ket{0}_A$}
    &
    & \rstick{$\ket{0}_A$}
    \\
    \lstick{$\ket{\psi}_S$}
    &
    & \rstick{$\ket{\psi}_S$}
\end{quantikz}

\caption{Ancilla preparation as an isometry.
}
\label{fig:fqp-ancilla-isometry}
\end{figure}
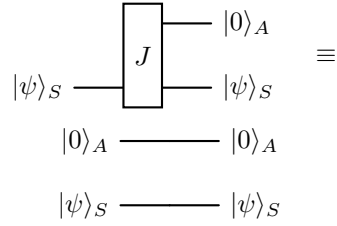

\begin{definition}[\((t,\delta,\eps)\)-success and query complexity in phase insensitive distance]\label{def:success}
A circuit is $(t,\delta,\eps)$-successful in fractional query problem if
\begin{equation}
 \sup_{\substack{H=H^\dagger\\\spec(H)\subset[\delta,2\pi)}}
  d_{ph}\bigl(\cA_{e^{iH}}J,Je^{itH}\bigr)\leq\eps.
 \label{eq:success-condition}
\end{equation}
Equivalently, for every promised $H$ there is a phase $\omega_H\in\T$ such that
\begin{equation}
 \opnorm{\cA_{e^{iH}}J-\omega_HJe^{itH}}\leq\eps.
 \label{eq:success-condition-phase}
\end{equation}
The phase may depend on $H$, but it must be common to all input states. The quantity $Q(t,\delta,\eps)$ is the smallest $N\in\N$ for which such a circuit exists. Figure~\ref{fig:fqp-coherent-success} illustrates this success criterion.
\end{definition}

By Lemma~\ref{lem:relationdiamond}, the success criterion in Definition~\ref{def:diamondsuccess} and in Definition~\ref{def:success} are equivalent up to universal constants.

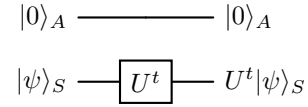
\begin{figure}[htbp]
\centering

\begin{minipage}[c]{0.43\linewidth}
\centering
\textbf{Implemented algorithm}
$\mathcal A_UJ$\\[12pt]

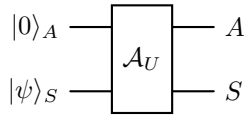
\begin{quantikz}[
    row sep={0.85cm,between origins},
    column sep=0.55cm
]
    \lstick{$\ket{0}_A$}
    & \gate[wires=2]{\mathcal A_U}
    & \rstick{$A$}
    \\
    \lstick{$\ket{\psi}_S$}
    &
    & \rstick{$S$}
\end{quantikz}
\end{minipage}
\hfill
\begin{minipage}[c]{0.08\linewidth}
\centering
$\displaystyle\simeq_{\varepsilon}$
\end{minipage}
\hfill
\begin{minipage}[c]{0.43\linewidth}
\centering
\textbf{Target}
$JU^t$\\[12pt]

\begin{quantikz}[
    row sep={0.85cm,between origins},
    column sep=0.55cm
]
    \lstick{$\ket{0}_A$}
    &
    & \rstick{$\ket{0}_A$}
    \\
    \lstick{$\ket{\psi}_S$}
    & \gate{U^t}
    & \rstick{$U^t\ket{\psi}_S$}
\end{quantikz}
\end{minipage}

\caption{Successful approximation. The symbol $\simeq_{\varepsilon}$ denotes operator norm error at most $\varepsilon$ after minimization over a global phase.
}
\label{fig:fqp-coherent-success}
\end{figure}

\section{QSVT upper bound construction in our model}
\label{app:upper}

In this section we present that the QSVT construction fits our model and yields a joint upper bound. Which we will later show is optimal in fractional query problem.

\begin{proof}[Proof of Theorem~\ref{thm:qsvt-upper}]

Since QSVT paper uses a slightly different definition of branch cut and its gap, we here show that our model includes their implementation under a shift.

Given \(U\) as defined in our setting, we put $\alpha=\pi+\delta/2$, $H'=H-\alpha I$, and $U'=e^{-i\alpha}U$. Then we have a shifted version of the unitary and the Hermitian that satisfies the QSVT's setting:
\begin{equation}
 U'=e^{iH'},\qquad
 \spec(H')\subset[-\pi+\delta/2,\pi-\delta/2),
 \qquad \opnorm{H'}\leq\pi-\delta/2.
\end{equation}

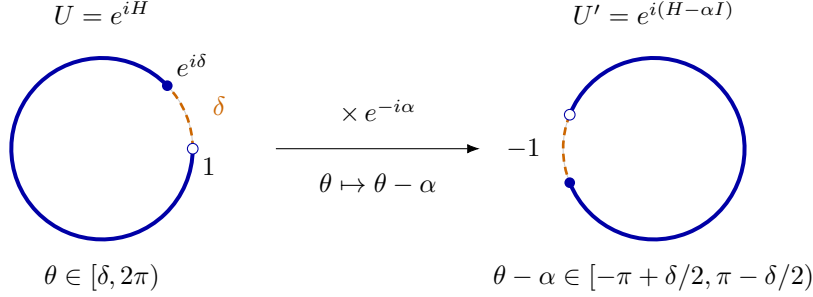
\begin{figure}[htbp]
\centering

\begin{tikzpicture}[
  font=\small,
  >=Latex,
  allowed/.style={blue!65!black,line width=1.5pt},
  gap/.style={orange!80!black,densely dashed,line width=1pt}
]

\begin{scope}[shift={(0,0)}]
  \draw[black!20] (0,0) circle (1.2);
  \draw[allowed]
    (44:1.2) arc[start angle=44,end angle=360,radius=1.2];
  \draw[gap]
    (0:1.2) arc[start angle=0,end angle=44,radius=1.2];

  \fill[blue!65!black] (44:1.2) circle (2pt);
  \draw[blue!65!black,fill=white] (0:1.2) circle (2pt);

  \node[above right] at (44:1.2) {$e^{i\delta}$};
  \node[below right] at (0:1.2) {$1$};
  \node[orange!80!black] at (20:1.65) {$\delta$};

  \node at (0,1.8) {$U=e^{iH}$};
  \node at (0,-1.7) {$\theta\in[\delta,2\pi)$};
\end{scope}

\draw[->] (2.3,0)--(5.0,0);
\node[above=6pt] at (3.65,0) {$\times\,e^{-i\alpha}$};
\node[below=6pt] at (3.65,0)
  {$\theta\mapsto\theta-\alpha$};

\begin{scope}[shift={(7.3,0)}]
  \draw[black!20] (0,0) circle (1.2);
  \draw[allowed]
    (-158:1.2) arc[start angle=-158,end angle=158,radius=1.2];
  \draw[gap]
    (158:1.2) arc[start angle=158,end angle=202,radius=1.2];

  \fill[blue!65!black] (-158:1.2) circle (2pt);
  \draw[blue!65!black,fill=white] (158:1.2) circle (2pt);

  \node[left=5pt] at (-1.2,0) {$-1$};
  \node at (0,1.8) {$U'=e^{i(H-\alpha I)}$};
  \node at (0,-1.7)
    {$\theta-\alpha\in[-\pi+\delta/2,\pi-\delta/2)$};
\end{scope}
\end{tikzpicture}

\caption{
The shift centers the allowed eigenphase interval at zero. Solid arcs indicate allowed eigenphases; dashed arcs indicate the excluded gap of width $\delta$. The operators $U'$ and $U$ agree up to global phase and induce the same unconditioned channel. Their fractional targets also differ only by a common phase.
}
\label{fig:qsvt-spectral-shift}
\end{figure}

For given parameter \((t,\delta,\eps)\in(0,1)\), QSVT fractional query construction \cite[Corollary~72]{GSLW19} and its following extension gives a unitary circuit $\cA_U$ made of controlled-\(U'\) and its inverse satisfying
\begin{equation}
 \opnorm{J^\dagger\cA_UJ-e^{itH'}}\leq\frac{\eps^2}{2},
 \qquad
 N=O\!\left(\frac1\delta\log\frac1\eps\right).
\end{equation}
This includes the amplification to a unit normalized block and the uncomputation of phase estimation. The circuit uses $N$ controlled- \((U')^{\pm1}\) queries and all other gates depend only on $(t,\delta,\eps)$.

For either query sign,
\begin{equation}
 \cO_{U'}^{(\sigma)}
 =\bigl(e^{-i\sigma\alpha}\ket0\!\bra0_C
          +\ket1\!\bra1_C\bigr)\cO_U^{(\sigma)},
 \qquad \sigma\in\{-1,+1\}.
\end{equation}

Each shifted query costs one allowed query. Absorbing the known phase gates and control register swaps into the adjacent query independent $W_j$ gives our query model.

\begin{figure}[htbp]
\centering

\begin{quantikz}[row sep=0.4cm,column sep=0.4cm]
 \lstick{$C_j$} & \octrl{2}          & \qw \\
 \lstick{$C$}   & \qw               & \qw \\
 \lstick{$S$}   & \gate{(U')^\sigma} & \qw
\end{quantikz}
\quad$=$\quad
\begin{quantikz}[row sep=0.4cm,column sep=0.4cm]
 \lstick{$C_j$} & \swap{1} & \qw             & \qw             & \swap{1} & \qw \\
 \lstick{$C$}   & \targX{} & \octrl{1}       & \gate{P_\sigma} & \targX{} & \qw \\
 \lstick{$S$}   & \qw      & \gate{U^\sigma} & \qw             & \qw      & \qw
\end{quantikz}

\caption{
Implementation of a shifted query. Where \(P_\sigma
 =e^{-i\sigma\alpha}|0\rangle\!\langle0|
   +|1\rangle\!\langle1|,
\) and \(\sigma\in\{-1,+1\}\).}
\label{fig:qsvt-query-transfer}
\end{figure}
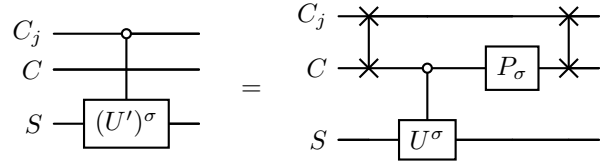

To check that this algorithm satisfies the \((t,\delta,\eps)\)-success criterion, let $\ket\psi$ be any unit vector. Unitarity gives
\begin{equation}
\begin{aligned}
     \bigl\|\cA_UJ\ket\psi-Je^{itH'}\ket\psi\bigr\|_2^2
     &=2-2\operatorname{Re}
       \bra\psi e^{-itH'}J^\dagger\cA_UJ\ket\psi\\
     &\leq 2\opnorm{J^\dagger\cA_UJ-e^{itH'}}
     \leq\eps^2.
\end{aligned}
\end{equation}
Since $e^{itH'}=e^{-it\alpha}U^t$, we conclude that
\begin{equation}
 d_{ph}(\cA_UJ,JU^t)
 \leq\opnorm{\cA_UJ-e^{-it\alpha}JU^t}\leq\eps.
\end{equation}
The estimate is uniform over all promised $H$.
\end{proof}

\section{Proofs of the main results}
\label{sec:lower_bound_proof}
We first use a hard subfamily to reduce the fractional query problem to a scalar polynomial approximation problem, then prove the \(\Omega_{\tau}\left(\frac{1}{\delta}log\frac{1}{\eps}\right)\) lower bound of the fractional query problem. Therefore, we conclude that the optimal query complexity for fractional query problem is \(\Theta_{\tau}\left(\frac{1}{\delta}log\frac{1}{\eps}\right)\). In the end of the section we also present a proof for \(\Omega_{\tau}\left(log\frac{1}{\eps}\right)\) lower bound, uniform to \(\delta \in (0,1)\), with a different method.

\subsection{Reduction to a scalar polynomial approximation problem}

Quantum circuits constructed with finite query of the variations of one single unitary or a fixed set of unitaries have some kind of finite frequency polynomial structure. This idea is investigated under many settings, spanning almost all the field~\cite{Beals2001Quantum,LowChuang2017Optimal,GSLW19}. Problems involving approximation with quantum circuits can sometimes be transformed into classical polynomial approximation problem, and we could often pick up various tools from classical approximation theory. QSP based methods specify a form of circuit where the polynomial structure is explicit and use this viewpoint to construct circuits for specific tasks, and the analysis revolving around it provided many informing results regarding query complexity of quantum algorithms ~\cite{LowChuang2017Optimal,Haah2019Product,MW24}. 

For our purpose, we reduce the approximation problem by considering a hard subfamily where there is exploitable polynomial structure and provide mechanisms obtaining lower bounds. The essential obstruction is the rigidity of the finite frequency set imposed by finitely many oracle calls. We first reduce from the approximation problem over the oracle family \(U=e^{iH}\), where \(H=H^\dagger\) and \(\spec(H)\subset[\delta,2\pi)\) to a hard subfamily indexed by a parameter \(\theta\) that has polynomial structure in this parameter and require more queries. If an algorithm success uniformly for all oracle family in Definition~\ref{def:fractional-query}, then it has to success at least for this hard subfamily.

We construct this hard subfamily artificially. Consider the following two eigenphase family with one fixed reference phase. Fix a projection \(P=\ket1\!\bra1\) and a unit vector \(\ket0\perp\ket1\), and set
\begin{equation}
  H_{\theta}=\delta I_S+(\theta-\delta)P,  \qquad
  \theta\in[\delta,2\pi).
\end{equation}
So that
\begin{equation}
  U_{\theta}= e^{iH_{\theta}}=e^{i\delta}(I_S-P)+e^{i\theta}P,
  \qquad
  \theta\in[\delta,2\pi).
  \label{eq:scalar-instances}
\end{equation}
This family satisfies all promises defined in Definition~\ref{def:fractional-query} because \(H_\theta=diag(\delta,\theta,\delta, \cdots,\delta)\). Since the original problem requires uniform success over the whole oracle family, it is at least as hard as this restriction.

\begin{lemma}[Reduction to a scalar Laurent polynomial approximation]
\label{lem:query-polynomial}
Let \(\cA_{U_{\theta}}\) be an \((t,\delta,\eps)\)-success \(N\)-query circuit satisfying Eq.~\eqref{eq:success-condition} for approximating \(U_\theta\). Then there exists a Laurent polynomial
\begin{equation}
  p(z)=\sum_{k=-2N}^{2N}a_kz^k
  \label{eq:p-polynomial}
\end{equation}
such that
\begin{equation}
  \abs{p(e^{i\theta})-e^{it\theta}}
  \leq2\eps,
  \qquad
  \theta\in[\delta,2\pi).
  \label{eq:scalar-approximation}
\end{equation}
\end{lemma}

\phantomsection\label{proof:reduction}
\begin{proof}[Proof for Lemma~\ref{lem:query-polynomial}]
Let \(z=e^{i\theta}\), and consider the two eigenphase family in Eq.~\eqref{eq:scalar-instances}. We want to write the algorithm \(\cA_{U_{\theta}}\) as a operator valued Laurent polynomial in \(z\). We first decompose the controlled oracle into a sum of \(z\)-dependent term and \(z\)-independent term. The \(j\)-th controlled query has the form
\begin{equation}
  \mathcal O_{U_{\theta}}^{(\sigma_j)}
  =z^{\sigma_j}Q_1+Q_0^{(\sigma_j)},
\end{equation}
where
\begin{equation}
  Q_1=\ket0\!\bra0_C\otimes P,
  \qquad
  Q_0^{(\sigma_j)}
  =\ket0\!\bra0_C\otimes e^{i\sigma_j\delta}(I_S-P)
   +\ket1\!\bra1_C\otimes I_S,
\end{equation}
and both operators are independent of \(z\). Every monomial obtained by expanding Eq.~\eqref{eq:n-query-circuit} has exponent equal to a sum of at most \(N\) numbers from \(\{-1,+1\}\), and hence lies between \(-N\) and \(N\). Therefore
\begin{equation}
  \cA_{U_{\theta}}=\cA(z)=\sum_{k=-N}^{N}z^kB_k,
  \label{eq:operator-polynomial}
\end{equation}
with \(z\)-independent operator coefficients \(B_k\).

Write \(\ket{v_j}=J\ket j\) for \(j\in\{0,1\}\), and define
\begin{equation}
    a(z):=\langle v_0\vert\cA(z)\vert v_0\rangle,
    \qquad
    b(z):=\langle v_1\vert\cA(z)\vert v_1\rangle.
\end{equation}
Both \(a\) and \(b\) have Laurent support contained in \(\{-N,\ldots,N\}\). By success, for every \(\theta\in[\delta,2\pi)\) there exists \(\omega_\theta\in\mathbb T\) such that
\begin{equation}
    \abs{a(e^{i\theta})-\omega_\theta e^{it\delta}}\leq\eps,
    \qquad
    \abs{b(e^{i\theta})-\omega_\theta e^{it\theta}}\leq\eps.
    \label{eq:two-amplitude-approximation}
\end{equation}
Since
\begin{equation}
\begin{aligned}
\left|a(e^{i\theta})-\omega_\theta e^{it\delta}\right|
&=\left|\langle v_0|E_\theta|0\rangle\right|
 \leq \|v_0\|_2\,\|E_\theta|0\rangle\|_2
 \leq\varepsilon,\\
\left|b(e^{i\theta})-\omega_\theta e^{it\theta}\right|
&=\left|\langle v_1|E_\theta|1\rangle\right|
 \leq \|v_1\|_2\,\|E_\theta|1\rangle\|_2
 \leq\varepsilon.
\end{aligned}
\end{equation}

If \(a(z)=\sum_{k=-N}^{N}\alpha_kz^k\), define the complex conjugate
\begin{equation}
    a^{\sharp}(z):=\sum_{k=-N}^{N}\overline{\alpha_k}z^{-k},
    \qquad
    p(z):=e^{it\delta}a^{\sharp}(z)b(z).
    \label{eq:phase-cancelled-polynomial}
\end{equation}
Then \(p\) has Laurent support in \(\{-2N,\ldots,2N\}\), and on the unit circle \(a^{\sharp}(z)=\overline{a(z)}\). Since \(\cA(z)\) is unitary for \(\abs z=1\), we have \(\abs{b(z)}\leq1\). Thus Eq.~\eqref{eq:two-amplitude-approximation} gives
\begin{equation}
\begin{aligned}
    \abs{p(e^{i\theta})-e^{it\theta}}
    &=
    \abs{\overline{a(e^{i\theta})}b(e^{i\theta})-e^{it(\theta-\delta)}}\\
    &\leq
    \abs{a(e^{i\theta})-\omega_\theta e^{it\delta}}\,
    \abs{b(e^{i\theta})}
    +
    \abs{b(e^{i\theta})-\omega_\theta e^{it\theta}}\\
    &\leq2\eps.
\end{aligned}
\end{equation}
The product in Eq.~\eqref{eq:phase-cancelled-polynomial} cancels the unknown common phase \(\omega_\theta\), proving Eq.~\eqref{eq:scalar-approximation}. Equivalently, \(\overline{a(z)}b(z)\) is a matrix coefficient of the induced unitary channel acting on \(\ket{v_1}\!\bra{v_0}\), so the global phase disappears and the Laurent degree doubles from \(N\) to \(2N\).
\end{proof}

Set
\begin{equation}
    q(\theta):=p(e^{i\theta})
    =\sum_{k=-2N}^{2N}a_ke^{ik\theta}.
    \label{eq:q-polynomial}
\end{equation}
It has at most \(4N+1\) integer frequencies, all contained in \(\{-2N,-2N+1,\ldots,2N\}\). Moreover, \(q\) is a trigonometric polynomial.

\subsection{Proof of the joint lower bound}

The idea of periodic mismatch lower bound is inherited from the periodic mismatch mechanism introduced by Liu and Pira for studying the depth error tradeoff of fixed data re-uploading circuits~\cite[Secs.~III.3 and IV.2]{LiuPira2026}. In both settings discussed here and there, approximation task is required on a smaller local interval, and the uniform approximation error is lower bounded by the obstruction over a larger interval by a Remez type inequality. The fundamental obstruction being used is that the target violates the periodicity of the trigonometric polynomial approximants.

We introduce the Sharp Remez inequality from \cite{TikhonovYuditskii2020SharpRemez}. It gives us the suitable tool to compare the smallness in the local region with the mismatch happened at periodicity.

\begin{theorem}[Sharp trigonometric Remez inequality, adapted from
{\cite[Theorem~1.1]{TikhonovYuditskii2020SharpRemez}}]
\label{thm:sharp-remez}
Let \(d\in\mathbb N\), \(0<s<2\pi\), and let
\[
    r(\theta)=\sum_{k=-d}^{d}c_ke^{ik\theta},
    \qquad c_k\in\mathbb C.
\]
For every measurable \(E\subset[0,2\pi)\) with
\(|E|\geq2\pi-s\),
\begin{equation}
    \sup_{\theta\in[0,2\pi)}|r(\theta)|
    \leq
    \mathsf T_{2d}\!\left(\sec\frac{s}{4}\right)
    \sup_{\theta\in E}|r(\theta)|,
    \label{eq:sharp-remez}
\end{equation}
where \(\mathsf T_m\) is the Chebyshev polynomial of the first kind,
characterized by \(\mathsf T_m(\cos u)=\cos(mu)\).
The constant is sharp.
\end{theorem}

\phantomsection\label{proof:remez}
\begin{proof}[Proof of Theorem~\ref{thm:remez-lower}]
By Lemma~\ref{lem:query-polynomial},
\(q(\theta):=p(e^{i\theta})\) is a trigonometric polynomial of degree
at most \(2N\) satisfying
\[
    |q(\theta)-e^{it\theta}|\leq2\eps,
    \qquad \delta\leq\theta<2\pi.
\]
Define
\begin{equation}
    r(\theta):=q(\theta+\delta)-e^{it\delta}q(\theta).
    \label{eq:remez-shifted-difference}
\end{equation}
This is again a trigonometric polynomial of degree at most \(2N\).
For \(\theta\in E:=[\delta,2\pi-\delta)\), both \(\theta\) and
\(\theta+\delta\) lie in the promised interval. Since
\(e^{it(\theta+\delta)}=e^{it\delta}e^{it\theta}\),
\begin{equation}
    \sup_{\theta\in E}|r(\theta)|\leq4\eps,
    \qquad |E|=2\pi-2\delta.
    \label{eq:remez-small-set}
\end{equation}

Continuity and periodicity of \(q\) imply
\[
    q(0)=q(2\pi),
    \qquad
    |q(0)-e^{2\pi it}|\leq2\eps,
\]
where the inequality follows by taking \(\theta\uparrow2\pi\).
No approximation promise at \(\theta=0\) is used. Together with
\(|q(\delta)-e^{it\delta}|\leq2\eps\), this gives
\begin{equation}
    \begin{aligned}
        |r(0)|
        &=|q(\delta)-e^{it\delta}q(0)|\\
        &\geq|e^{it\delta}(1-e^{2\pi it})|-4\eps\\
        &=2\sin(\pi t)-4\eps.
    \end{aligned}
    \label{eq:remez-branch-defect}
\end{equation}
Apply Theorem~\ref{thm:sharp-remez} with \(d=2N\) and \(s=2\delta\).
Equations~\eqref{eq:remez-small-set} and
\eqref{eq:remez-branch-defect} yield
\[
    2\sin(\pi t)-4\eps
    \leq |r(0)|
    \leq4\eps\,\mathsf T_{4N}\!\left(\sec\frac{\delta}{2}\right).
\]
Rearranging, we have

\begin{equation}
    \eps
    \geq
    \frac{\sin(\pi t)}
    {2\bigl(1+\mathsf T_{4N}(\sec(\delta/2))\bigr)}
    \label{eq:remez first part}
\end{equation}
To obtain the exponential form, put
\begin{equation}
    \begin{aligned}
        a_\delta
        &:=\operatorname{arcosh}\!\left(\sec\frac{\delta}{2}\right)\\
        &=\log\!\left(\sec\frac{\delta}{2}
                     +\tan\frac{\delta}{2}\right)
          =\int_0^{\delta/2}\sec u\,du
          \leq\delta,
    \end{aligned}
    \label{eq:remez-exponent}
\end{equation}
since \(0<\delta\leq1\) and \(\sec u\leq2\) on \([0,1/2]\).
Hence
\[
    \mathsf T_{4N}\!\left(\sec\frac{\delta}{2}\right)
    =\cosh(4Na_\delta)
    \leq e^{4N\delta}.
\]
Using \(1+e^{4N\delta}\leq2e^{4N\delta}\), we have
\begin{equation}
    \frac{\sin(\pi t)}
    {2\bigl(1+\mathsf T_{4N}(\sec(\delta/2))\bigr)}
    \geq
    \frac{\sin(\pi t)}{4}e^{-4N\delta}.
    \label{eqn:remezsecondpart}
\end{equation}
Combining Eq.(\ref{eq:remez first part}) and Eq.(\ref{eqn:remezsecondpart}), rearranging and taking logarithm we have
\begin{equation}
    N\geq\frac{1}{4\delta}
    \log\frac{\sin(\pi t)}{4\eps}.
    \label{eq:remez-query-lower}
\end{equation}

Finally, let \(s_\tau:=\sin(\pi\tau)\). If
\(t\in[\tau,1-\tau]\) and \(\eps\leq(s_\tau/4)^2\), then
\[
    \log\frac{\sin(\pi t)}{4\eps}
    \geq\log\frac{s_\tau}{4\eps}
    \geq\frac12\log\frac1\eps.
\]
Substitution into Eq.~\eqref{eq:remez-query-lower}, followed by
minimization over successful circuits, proves
Eq.~\eqref{eq:remez-uniform-lower}.
\end{proof}

\subsection{Proof of the optimal complexity}
\phantomsection\label{app:main}
\begin{proof}[Proof for Theorem~\ref{thm:optimality}]
Let an $N$-query circuit be $(t,\delta,\eps)$-successful.
By Theorem~\ref{thm:remez-lower},
\[
    N\geq\frac{1}{4\delta}
        \log\frac{\sin(\pi t)}{4\eps}.
\]
Since $t\in[\tau,1-\tau]$, we have $\sin(\pi t)\geq\sin(\pi\tau)$. Moreover, $\sin(\pi\tau)/4\geq\sqrt{\eps}$, and hence
\[
    \log\frac{\sin(\pi t)}{4\eps}
    \geq
    \log\frac{\sin(\pi\tau)}{4\eps}
    \geq\frac12\log\frac{1}{\eps}.
\]
Therefore every successful circuit satisfies
\[
    N\geq\frac{1}{8\delta}\log\frac{1}{\eps}.
\]
Minimizing over such circuits proves the lower bound in
Eq.~\eqref{eq:optimal-query-complexity}.

For the upper bound, observe that $\left(\frac{\sin(\pi\tau)}{4}\right)^2\leq1/16<1/2$.
Theorem~\ref{thm:qsvt-upper} therefore supplies a
$(t,\delta,\eps)$-successful circuit using at most
$C\frac{1}{\delta}\log(\frac{1}{\eps})$ queries, for a universal constant $C$.
Combining the two bounds proves the result.
\end{proof}

\section{Alternative proof for \(\eps\) lower bound}

We introduce a \(\Omega_{\tau}\left(log\frac{1}{\eps}\right)\)lower bound proof based on functional norm inequality. We will construct a normalized functional that annihilates the attainable frequency space but not the target, and use the functional norm inequality to estimate the error. 

All accessible frequencies lie in the finite set \(\{-2N,\dots,2N\}\), whose exponentials span a finite dimensional subspace of \(C(I)\). A functional that vanishes on this subspace but not on the target gives a lower bound on their uniform distance. 

\begin{theorem}[Functional norm inequality]
    A standard way to lower bound a linear approximation problem is to let \(\mathcal{F}\) be the closure of the span of the approximants and construct, for \(f\notin\mathcal F\), a nonzero linear functional \(\phi\) such that
    \begin{equation}
        \phi(f)\neq0,
        \qquad \phi(g)=0\quad\text{for every }g\in\mathcal F.
    \end{equation}
    Hahn--Banach theorem guarantees the existence of such a functional. Its norm inequality gives
    \begin{equation}
        \abs{\phi(f-g)}\leq\norm{\phi}\norm{f-g}_\infty
    \end{equation}
    so that
    \begin{equation}
        \norm{f-g}_\infty \geq \frac{\abs{\phi(f)}}{\norm{\phi}}
    \end{equation}
\end{theorem}

Our proof constructs this functional explicitly as a normalized finite frequency annihilator that encode the query count implicitly, and then uses elementary estimation to obtain a lower bound with respect to \(\eps\).

\subsection{Constructing a frequency annihilator}

For \(h\in\mathbb R\), let \(T_h\) denote translation,
\begin{equation}
  (T_hf)(x)=f(x+h).
  \label{eq:translation}
\end{equation}
Every exponential mode is an eigenfunction:
\begin{equation}
  T_he^{i\lambda x}=e^{i\lambda h}e^{i\lambda x}.
  \label{eq:translation-eigen}
\end{equation}
This permits one polynomial of \(T_h\) to erase a set of frequencies by designing its roots.

\begin{theorem}[A frequency annihilator of the approximation space]
\label{thm:frequency annihilator-functional}
Let
\[
    \Lambda=\{\lambda_1,\ldots,\lambda_s\}\subset\mathbb R,
    \qquad
    \mathcal V_\Lambda
    :=
    \operatorname{span}
    \{e^{i\lambda x}:\lambda\in\Lambda\},
\]
where \(s=|\Lambda|\geq 1\). Choose \(h>0\) and \(x_0\in I\) such that
\[
    x_0,x_0+h,\ldots,x_0+sh\in I.
\]

Define the polynomial
\begin{equation}
    R(w)
    :=
    \prod_{\lambda\in\Lambda}
    \bigl(w-e^{i\lambda h}\bigr)
    =
    \sum_{j=0}^{s}r_jw^j,
    \label{eq:filter-polynomial}
\end{equation}
and set
\begin{equation}
    \rho_R:=\sum_{j=0}^{s}|r_j|.
    \label{eq:filter-rho}
\end{equation}
Replacing \(w\) by \(T_h\) gives the linear operator
\[
    R(T_h)
    =
    \prod_{\lambda\in\Lambda}
    \bigl(T_h-e^{i\lambda h}\mathrm{Id}\bigr)
    =
    \sum_{j=0}^{s}r_jT_h^j.
\]
Thus
\[
    (R(T_h)f)(x)
    =
    \sum_{j=0}^{s}r_jf(x+jh).
\]
For \(f\in C(I)\), we use this formula only at points \(x\) for which all the required translates \(x,x+h,\ldots,x+sh\) lie in \(I\).

Then
\[
    R(T_h)q=0,
    \qquad q\in\mathcal V_\Lambda,
\]
at every point where the expression is defined. Moreover,
\begin{equation}
    \varphi_R(f)
    :=
    \frac{(R(T_h)f)(x_0)}{\rho_R}
    =
    \frac{1}{\rho_R}
    \sum_{j=0}^{s}r_jf(x_0+jh)
    \label{eq:filter-functional}
\end{equation}
is a bounded linear functional on \(C(I)\) satisfying
\[
    \|\varphi_R\|=1,
    \qquad
    \varphi_R(q)=0
    \quad\text{for every }q\in\mathcal V_\Lambda.
\]
For every \(\xi\in\mathbb R\), its value on an exponential is
\[
    \varphi_R\!\left(e^{i\xi(\cdot)}\right)
    =
    \frac{e^{i\xi x_0}}{\rho_R}R(e^{i\xi h}).
\]
Consequently, for every \(f\in C(I)\),
\begin{equation}
    \begin{aligned}
        \inf_{q\in\mathcal V_\Lambda}\|f-q\|_\infty
        &\geq |\varphi_R(f)|\\
        &=
        \frac{|(R(T_h)f)(x_0)|}{\rho_R}
        =
        \frac{
            \left|\sum_{j=0}^{s}r_jf(x_0+jh)\right|
        }{
            \sum_{j=0}^{s}|r_j|
        }.
    \end{aligned}
    \label{eq:exact-filter-general}
\end{equation}
\end{theorem}

\begin{proof}
Since \(r_s=1\), we have \(\rho_R>0\).

The reason for our choice of \(R\) is that translation diagonalizes the functions of the form \(e^{i\xi x}\) ,
\[
    T_he^{i\xi(\cdot)}
    =
    e^{i\xi h}e^{i\xi(\cdot)}.
\]
More explicitly,
\[
    \begin{aligned}
        \bigl(R(T_h)e^{i\xi(\cdot)}\bigr)(x)
        &=
        \sum_{j=0}^{s}r_je^{i\xi(x+jh)}\\
        &=
        e^{i\xi x}
        \sum_{j=0}^{s}r_j(e^{i\xi h})^j\\
        &=
        e^{i\xi x}R(e^{i\xi h}).
    \end{aligned}
\]
When \(\xi\in\Lambda\), the last expression is zero because \(e^{i\xi h}\) is a root of \(R\). By linearity, \(R(T_h)\) therefore sends every \(q\in\mathcal V_\Lambda\) to zero. Evaluating at \(x_0\) and dividing by \(\rho_R\) also gives the stated formula for \(\varphi_R\) on exponentials.

The formula defining \(\varphi_R\) is a finite linear combination of point evaluations, so \(\varphi_R\) is linear. The triangle inequality gives
\[
    \begin{aligned}
        |\varphi_R(f)|
        &\leq
        \frac{1}{\rho_R}
        \sum_{j=0}^{s}|r_j|\,|f(x_0+jh)|\\
        &\leq
        \frac{1}{\rho_R}
        \sum_{j=0}^{s}|r_j|\,\|f\|_\infty\\
        &=
        \|f\|_\infty.
    \end{aligned}
\]
Hence \(\|\varphi_R\|\leq 1\).

To show that the norm is exactly one, choose a continuous function \(g\) with
\[
    g(x_0+jh)
    =
    \begin{cases}
        \overline{r_j}/|r_j|, & r_j\neq 0,\\
        0, & r_j=0.
    \end{cases}
\]
Such a function with \(\|g\|_\infty\leq 1\) can be constructed by linear interpolation between consecutive sample points and constant extension to the rest of \(I\). The interpolated values stay in the closed unit disk. For this function,
\[
    \varphi_R(g)
    =
    \frac{1}{\rho_R}\sum_{j=0}^{s}|r_j|
    =
    1.
\]
Therefore \(\|\varphi_R\|=1\).

Finally, since \(\varphi_R(q)=0\) for every
\(q\in\mathcal V_\Lambda\),
\[
    |\varphi_R(f)|
    =
    |\varphi_R(f-q)|
    \leq
    \|f-q\|_\infty.
\]
Taking the infimum over \(q\in\mathcal V_\Lambda\) proves
\eqref{eq:exact-filter-general}.
\end{proof}

\subsection{\(\eps\) lower bound from frequency annihilator}

\phantomsection\label{proof:freqlower}

\begin{proof}[Proof of Theorem~\ref{thm:fd-lower}]

We want to show that every \((t,\delta,\eps)\)-successful \(N\)-query circuit satisfies
\begin{equation}
    \eps
    \geq
    \frac12
    \prod_{k=-2N}^{2N}
    \sin\!\left(
        \frac{\pi\abs{k-t}}{2(4N+1)}
    \right).
    \label{eq:fd-product}
\end{equation}
For every \(N\in\mathbb{N}\), this implies
\begin{equation}
    \eps
    \geq
    \frac{t(1-t)}{2}(6e)^{-4N}.
    \label{eq:fd-geometric}
\end{equation}
Consequently,
\begin{equation}
    N 
    \geq
    \max\!\left\{
        0,
        \left\lceil
            \frac{\log(t(1-t)/(2\eps))}{4\log(6e)}
        \right\rceil
    \right\}.
\end{equation}

Let \(q\) be the Laurent polynomial from Lemma~\ref{lem:query-polynomial}. Equation~\eqref{eq:scalar-approximation} implies
\begin{equation}
    \abs{q(\theta)-e^{it\theta}}
    \leq2\eps,
    \qquad
    \theta\in[\delta,2\pi).
    \label{eq:shifted-approximation}
\end{equation}
Define
\begin{equation}
    h:=\frac{\pi}{4N+1},
    \qquad
    (T_hg)(\theta):=g(\theta+h),
\end{equation}
and let
\begin{equation}
    R(w):=\prod_{k=-2N}^{2N}\bigl(w-e^{ikh}\bigr)
    =\sum_{j=0}^{4N+1}r_jw^j.
    \label{eq:annihilator-poly}
\end{equation}
Each mode \(e^{ik\theta}\) is an eigenfunction of \(T_h\) with eigenvalue \(e^{ikh}\). Hence \(R\) acting on \(T_h\) defines a frequency annihilator
\begin{equation}
    R(T_h)q=0.
    \label{eq:annihilate-q}
\end{equation}
On the target,
\begin{equation}
    R(T_h)e^{it\theta}
    =R(e^{ith})e^{it\theta}.
    \label{eq:target-eigenmode}
\end{equation}

Put \(F=q-e^{it\theta}\). Since
\begin{equation}
    \delta+(4N+1)h=\delta+\pi<2\pi,
\end{equation}
all points evaluated by the frequency annihilator \(\delta,\delta+h,\ldots,\delta+(4N+1)h\) lie in the promised interval. Evaluating
\(
R(T_h) F=-R(e^{ith})e^{it\theta}
\)
at \(\theta=\delta\) gives
\begin{equation}
\begin{aligned}
    \abs{R(e^{ith})}
    &\leq
    \sum_{j=0}^{4N+1}\abs{r_j}\abs{ F(\delta+jh)}\\
    &\leq
    2\eps\sum_{j=0}^{4N+1}\abs{r_j}.
\end{aligned}
\label{eq:fd-l1-step}
\end{equation}
Because every root of \(R\) has modulus one, the coefficient \(\ell^1\)-norm satisfies
\begin{equation}
    \sum_{j=0}^{4N+1}\abs{r_j}
    \leq
    \prod_{k=-2N}^{2N}\bigl(1+\abs{e^{ikh}}\bigr)
    =2^{4N+1}.
    \label{eq:R-l1-bound}
\end{equation}
Moreover,
\begin{equation}
\begin{aligned}
    \abs{R(e^{ith})}
    &=
    \prod_{k=-2N}^{2N}\abs{e^{ith}-e^{ikh}}\\
    &=
    2^{4N+1}
    \prod_{k=-2N}^{2N}
    \sin\!\left(
      \frac{\pi\abs{k-t}}{2(4N+1)}
    \right).
\end{aligned}
\label{eq:R-target-product}
\end{equation}
Combining Eqs.~\eqref{eq:fd-l1-step}--\eqref{eq:R-target-product} proves Eq.~\eqref{eq:fd-product}.

Assume first that \(N\geq1\), and put \(M:=2N\). Every sine argument in Eq.~\eqref{eq:fd-product} belongs to \([0,\pi/2]\), so \(\sin x\geq2x/\pi\) gives
\begin{equation}
\begin{aligned}
    2\eps
    &\geq
    \frac{\prod_{k=-M}^{M}\abs{k-t}}{(2M+1)^{2M+1}}\\
    &\geq
    t(1-t)\frac{M!(M-1)!}{(2M+1)^{2M+1}}.
\end{aligned}
\label{eq:fd-factorial}
\end{equation}
Indeed, the factors at \(k=0\) and \(k=1\) are \(t\) and \(1-t\), the negative index factors have product at least \(M!\), and the remaining positive index factors have product at least \((M-1)!\). Using \(M!\geq(M/e)^M\), \(M/(2M+1)\geq1/3\), and \(M(2M+1)\leq4^M\), we obtain
\begin{equation}
\begin{aligned}
    \frac{M!(M-1)!}{(2M+1)^{2M+1}}
    &=
    \frac{(M!)^2}{M(2M+1)^{2M+1}}\\
    &\geq
    \frac{1}{e^{2M}M(2M+1)}
    \left(\frac{M}{2M+1}\right)^{2M}\\
    &\geq
    (6e)^{-2M}.
\end{aligned}
\end{equation}
Since \(M=2N\), this proves Eq.~\eqref{eq:fd-geometric} for \(N\geq1\). For \(N=0\), Eq.~\eqref{eq:fd-product} reads \(\eps\geq\tfrac12\sin(\pi t/2)\), and \(\sin(\pi t/2)\geq t\geq t(1-t)\). Thus the same geometric bound holds for all \(N\in\mathbb{N}\). Solving it for \(N\) proves Eq.~\eqref{eq:fd-query-lower}.
\end{proof}

\bibliography{references}

@article{Beals2001Quantum,
  author  = {Robert Beals and Harry Buhrman and Richard Cleve and Michele Mosca and Ronald de Wolf},
  title   = {Quantum Lower Bounds by Polynomials},
  journal = {Journal of the ACM},
  volume  = {48},
  number  = {4},
  pages   = {778--797},
  year    = {2001},
  doi     = {10.1145/502090.502097}
}

@article{Regula2021operational,
  doi = {10.22331/q-2021-08-09-522},
  url = {https://doi.org/10.22331/q-2021-08-09-522},
  title = {Operational applications of the diamond norm and related measures in quantifying the non-physicality of quantum maps},
  author = {Regula, Bartosz and Takagi, Ryuji and Gu, Mile},
  journal = {{Quantum}},
  issn = {2521-327X},
  publisher = {{Verein zur F{\"{o}}rderung des Open Access Publizierens in den Quantenwissenschaften}},
  volume = {5},
  pages = {522},
  month = aug,
  year = {2021}
}

@article{Haah2019Product,
  author  = {Jeongwan Haah},
  title   = {Product Decomposition of Periodic Functions in Quantum Signal Processing},
  journal = {Quantum},
  volume  = {3},
  pages   = {190},
  year    = {2019},
  doi     = {10.22331/q-2019-10-07-190},
  eprint  = {1806.10236},
  archivePrefix = {arXiv}
}

@inproceedings{HaahKothariODonnellTang2023,
   title={Query-optimal estimation of unitary channels in diamond distance},
   url={http://dx.doi.org/10.1109/FOCS57990.2023.00028},
   DOI={10.1109/focs57990.2023.00028},
   booktitle={2023 IEEE 64th Annual Symposium on Foundations of Computer Science (FOCS)},
   publisher={IEEE},
   author={Haah, Jeongwan and Kothari, Robin and O’Donnell, Ryan and Tang, Ewin},
   year={2023},
   month=Nov, pages={363–390}
}

@article{TikhonovYuditskii2020SharpRemez,
  author  = {Sergey Tikhonov and Peter Yuditskii},
  title   = {Sharp Remez Inequality},
  journal = {Constructive Approximation},
  volume  = {52},
  number  = {2},
  pages   = {233--246},
  year    = {2020},
  doi     = {10.1007/s00365-019-09473-2},
  eprint  = {1809.09726},
  archivePrefix = {arXiv}
}

@article{LowChuang2017Optimal,
  author  = {Guang Hao Low and Isaac L. Chuang},
  title   = {Optimal Hamiltonian Simulation by Quantum Signal Processing},
  journal = {Physical Review Letters},
  volume  = {118},
  pages   = {010501},
  year    = {2017},
  doi     = {10.1103/PhysRevLett.118.010501},
  eprint  = {1606.02685},
  archivePrefix = {arXiv}
}

@article{LYC16,
  author  = {Low, Guang Hao and Yoder, Theodore J. and Chuang, Isaac L.},
  title   = {Methodology of Resonant Equiangular Composite Quantum Gates},
  journal = {Physical Review X},
  volume  = {6},
  pages   = {041067},
  year    = {2016},
  doi     = {10.1103/PhysRevX.6.041067}
}

@article{LC19,
  author  = {Low, Guang Hao and Chuang, Isaac L.},
  title   = {{H}amiltonian Simulation by Qubitization},
  journal = {Quantum},
  volume  = {3},
  pages   = {163},
  year    = {2019},
  doi     = {10.22331/q-2019-07-12-163}
}

@inproceedings{GSLW19,
  author    = {Gily{\'e}n, Andr{\'a}s and Su, Yuan and Low, Guang Hao and Wiebe, Nathan},
  title     = {Quantum Singular Value Transformation and Beyond: Exponential Improvements for Quantum Matrix Arithmetics},
  booktitle = {Proceedings of the 51st Annual ACM SIGACT Symposium on Theory of Computing (STOC)},
  pages     = {193--204},
  year      = {2019},
  doi       = {10.1145/3313276.3316366},
  eprint    = {1806.01838},
  archivePrefix = {arXiv},
  note      = {Full version: arXiv:1806.01838}
}

@article{MRTC21,
  author  = {Martyn, John M. and Rossi, Zane M. and Tan, Andrew K. and Chuang, Isaac L.},
  title   = {Grand Unification of Quantum Algorithms},
  journal = {PRX Quantum},
  volume  = {2},
  pages   = {040203},
  year    = {2021},
  doi     = {10.1103/PRXQuantum.2.040203},
  eprint  = {2105.02859},
  archivePrefix = {arXiv}
}

@article{DLT22,
  author  = {Dong, Yulong and Lin, Lin and Tong, Yu},
  title   = {Ground-State Preparation and Energy Estimation on Early Fault-Tolerant Quantum Computers via Quantum Eigenvalue Transformation of Unitary Matrices},
  journal = {PRX Quantum},
  volume  = {3},
  pages   = {040305},
  year    = {2022},
  doi     = {10.1103/PRXQuantum.3.040305}
}

@article{MW24,
  author  = {Motlagh, Danial and Wiebe, Nathan},
  title   = {Generalized Quantum Signal Processing},
  journal = {PRX Quantum},
  volume  = {5},
  pages   = {020368},
  year    = {2024},
  doi     = {10.1103/PRXQuantum.5.020368},
  eprint  = {2308.01501},
  archivePrefix = {arXiv}
}

@misc{Aar06,
  author       = {Aaronson, Scott},
  title        = {The Ten Most Annoying Questions in Quantum Computing},
  howpublished = {Shtetl-Optimized (blog), \url{https://scottaaronson.blog/?p=112}},
  year         = {2006}
}

@article{SMM09,
  author  = {Sheridan, Lana and Maslov, Dmitri and Mosca, Michele},
  title   = {Approximating Fractional Time Quantum Evolution},
  journal = {Journal of Physics A: Mathematical and Theoretical},
  volume  = {42},
  number  = {18},
  pages   = {185302},
  year    = {2009},
  doi     = {10.1088/1751-8113/42/18/185302},
  eprint  = {0810.3843},
  archivePrefix = {arXiv}
}

@inproceedings{CGMSY09,
  author    = {Cleve, Richard and Gottesman, Daniel and Mosca, Michele and Somma, Rolando D. and Yonge-Mallo, David},
  title     = {Efficient Discrete-Time Simulations of Continuous-Time Quantum Query Algorithms},
  booktitle = {Proceedings of the 41st Annual ACM Symposium on Theory of Computing (STOC)},
  pages     = {409--416},
  year      = {2009},
  doi       = {10.1145/1536414.1536471},
  eprint    = {0811.4428},
  archivePrefix = {arXiv}
}

@inproceedings{BCCKS14,
  author    = {Berry, Dominic W. and Childs, Andrew M. and Cleve, Richard and Kothari, Robin and Somma, Rolando D.},
  title     = {Exponential Improvement in Precision for Simulating Sparse {H}amiltonians},
  booktitle = {Proceedings of the 46th Annual ACM Symposium on Theory of Computing (STOC)},
  pages     = {283--292},
  year      = {2014},
  doi       = {10.1145/2591796.2591854},
  eprint    = {1312.1414},
  archivePrefix = {arXiv}
}

@article{BunKothariThaler2020Polynomial,
  author  = {Bun, Mark and Kothari, Robin and Thaler, Justin},
  title   = {The Polynomial Method Strikes Back:
             Tight Quantum Query Bounds via Dual Polynomials},
  journal = {Theory of Computing},
  volume  = {16},
  number  = {10},
  pages   = {1--71},
  year    = {2020},
  doi     = {10.4086/toc.2020.v016a010}
}

@article{BACS07,
  author  = {Berry, Dominic W. and Ahokas, Graeme and Cleve, Richard and Sanders, Barry C.},
  title   = {Efficient Quantum Algorithms for Simulating Sparse {H}amiltonians},
  journal = {Communications in Mathematical Physics},
  volume  = {270},
  number  = {2},
  pages   = {359--371},
  year    = {2007},
  doi     = {10.1007/s00220-006-0150-x}
}

@inproceedings{BCK15,
  author    = {Berry, Dominic W. and Childs, Andrew M. and Kothari, Robin},
  title     = {{H}amiltonian Simulation with Nearly Optimal Dependence on All Parameters},
  booktitle = {IEEE 56th Annual Symposium on Foundations of Computer Science (FOCS)},
  pages     = {792--809},
  year      = {2015}
}

@inproceedings{MdW23,
  author    = {Mande, Nikhil S. and de Wolf, Ronald},
  title     = {Tight Bounds for Quantum Phase Estimation and Related Problems},
  booktitle = {31st Annual European Symposium on Algorithms (ESA 2023)},
  series    = {Leibniz International Proceedings in Informatics (LIPIcs)},
  volume    = {274},
  pages     = {81},
  year      = {2023},
  doi       = {10.4230/LIPIcs.ESA.2023.81}
}

@inproceedings{MS24,
  author    = {Montanaro, Ashley and Shao, Changpeng},
  title     = {Quantum and Classical Query Complexities of Functions of Matrices},
  booktitle = {Proceedings of the 56th Annual ACM Symposium on Theory of Computing (STOC)},
  pages     = {573--584},
  year      = {2024}
}

@article{AFCB14,
  author  = {Ara{\'u}jo, Mateus and Feix, Adrien and Costa, Fabio and Brukner, {\v{C}}aslav},
  title   = {Quantum Circuits Cannot Control Unknown Operations},
  journal = {New Journal of Physics},
  volume  = {16},
  number  = {9},
  pages   = {093026},
  year    = {2014},
  doi     = {10.1088/1367-2630/16/9/093026}
}

@article{QDSSM19a,
  author  = {Quintino, Marco T{\'u}lio and Dong, Qingxiuxiong and Shimbo, Atsushi and Soeda, Akihito and Murao, Mio},
  title   = {Reversing Unknown Quantum Transformations: Universal Quantum Circuit for Inverting General Unitary Operations},
  journal = {Physical Review Letters},
  volume  = {123},
  number  = {21},
  pages   = {210502},
  year    = {2019},
  doi     = {10.1103/PhysRevLett.123.210502}
}

@article{QDSSM19b,
  author  = {Quintino, Marco T{\'u}lio and Dong, Qingxiuxiong and Shimbo, Atsushi and Soeda, Akihito and Murao, Mio},
  title   = {Probabilistic Exact Universal Quantum Circuits for Transforming Unitary Operations},
  journal = {Physical Review A},
  volume  = {100},
  number  = {6},
  pages   = {062339},
  year    = {2019},
  doi     = {10.1103/PhysRevA.100.062339}
}

@article{YSM23,
  author  = {Yoshida, Satoshi and Soeda, Akihito and Murao, Mio},
  title   = {Reversing Unknown Qubit-Unitary Operation, Deterministically and Exactly},
  journal = {Physical Review Letters},
  volume  = {131},
  number  = {12},
  pages   = {120602},
  year    = {2023},
  doi     = {10.1103/PhysRevLett.131.120602}
}

@misc{GST24,
  author        = {Gavorov{\'a}, Zuzana and Seidel, Matan and Touati, Yonathan},
  title         = {Topological Obstructions to Quantum Computation with Unitary Oracles},
  year          = {2024},
  eprint        = {2011.10031},
  archivePrefix = {arXiv},
  primaryClass  = {quant-ph},
  note          = {arXiv:2011.10031v3}
}

@article{OYM25,
  author  = {Odake, Tatsuki and Yoshida, Satoshi and Murao, Mio},
  title   = {Analytical Lower Bound on Query Complexity for Transformations of Unknown Unitary Operations},
  journal = {Physical Review Letters},
  volume  = {135},
  number  = {23},
  pages   = {230603},
  year    = {2025},
  doi     = {10.1103/drp2-rzzw},
  eprint  = {2405.07625},
  archivePrefix = {arXiv}
}

@inproceedings{CYZ26,
  author    = {Chen, Kean and Yu, Nengkun and Zhang, Zhicheng},
  title     = {Approximation Does Not Help in Quantum Unitary Time-Reversal},
  booktitle = {Proceedings of the 58th Annual ACM Symposium on Theory of Computing (STOC)},
  pages     = {1376--1387},
  year      = {2026},
  doi       = {10.1145/3798129.3800847}
}

@article{Chiribella2008QuantumComb,
  author        = {Chiribella, Giulio and D'Ariano, Giacomo Mauro
                   and Perinotti, Paolo},
  title         = {Quantum Circuit Architecture},
  journal       = {Physical Review Letters},
  volume        = {101},
  number        = {6},
  pages         = {060401},
  year          = {2008},
  doi           = {10.1103/PhysRevLett.101.060401},
  eprint        = {0712.1325},
  archivePrefix = {arXiv},
  primaryClass  = {quant-ph}
}

@misc{LiuPira2026,
  author        = {Liu, Anthony Yuezhang and Pira, Lirand{\"e}},
  title         = {The Cost of Removing Tunability in Quantum Data Re-uploading},
  year          = {2026},
  eprint        = {2606.25598},
  archivePrefix = {arXiv},
  primaryClass  = {quant-ph}
}

@misc{DNSM19,
  author        = {Dong, Qingxiuxiong and Nakayama, Shojun and Soeda, Akihito and Murao, Mio},
  title         = {Controlled Quantum Operations and Combs, and Their Applications to Universal Controllization of Divisible Unitary Operations},
  year          = {2019},
  eprint        = {1911.01645},
  archivePrefix = {arXiv},
  primaryClass  = {quant-ph}
}

@misc{Kit95,
  author        = {Kitaev, A. Yu.},
  title         = {Quantum Measurements and the Abelian Stabilizer Problem},
  year          = {1995},
  eprint        = {quant-ph/9511026},
  archivePrefix = {arXiv}
}

@article{OKSM24,
  author  = {Odake, Tatsuki and Kristj{\'a}nsson, Hl{\'e}r and Soeda, Akihito and Murao, Mio},
  title   = {Higher-Order Quantum Transformations of {H}amiltonian Dynamics},
  journal = {Physical Review Research},
  volume  = {6},
  pages   = {L012063},
  year    = {2024},
  doi     = {10.1103/PhysRevResearch.6.L012063}
}

\end{document}